\documentclass[letterpaper, 10 pt, conference]{ieeeconf}  

\IEEEoverridecommandlockouts                              
\usepackage{amsmath,amssymb,bm}

\usepackage{color, subfig}
\usepackage{graphicx}
\usepackage{epstopdf}
\usepackage{epsfig}
\usepackage{graphics} 
\usepackage{epstopdf}
\usepackage{enumerate}
\usepackage{caption}
\usepackage{amsmath} 
\usepackage{subfig,pgfplots}
\pgfplotsset{ 
  compat=newest, 
  colormap={redblue}{ rgb255(0cm)=(128,0,0); rgb255(.4cm)=(255,0,0);rgb255(.8cm)=(255,255,0); rgb255(1.2cm)=(100,255,0);rgb255(3cm)=(0,255,255);rgb255(5cm)=(0,0,180)},
  legend style =
  {font=\small \sffamily},
  label style = {font=\small\sffamily},
every tick label/.append style={font=\small}}
\usepackage{tikz}

\usetikzlibrary {positioning}

\def\beq{\begin{equation}}
\def\eeq{\end{equation}}

\newtheorem{proposition}{Proposition}

\newtheorem{assumption}{Assumption}

\newcommand{\ba}{\begin{array}}
\newcommand{\ea}{\end{array}}

\newcommand{\be}{\begin{equation}}
\newcommand{\ee}{\end{equation}}

\newcommand{\mc}{\mathcal}

\newcommand{\1}{\mathbf{1}}
\newcommand{\0}{\mathbf{0}}

\newcommand{\N}{\mathbb{N}}

\newcommand{\tup}[1]{\textup{#1}}

\newcommand{\bs}[1]{\boldsymbol{#1}}

\renewcommand{\span}{\text{span}}

\DeclareMathOperator*{\argmin}{argmin}

\def\N{\mathbb{N}}

\def\diag{{\rm diag}\,}
\def\col{{\rm col}\,}

\usepackage{mathtools}
\usepackage[acronym]{glossaries}
\newacronym{SIS}{SIS}{Susceptible--Infected--Susceptible}
\newacronym{NPI}{NPIs}{nonpharmaceutical interventions}
\newacronym{NLP}{NLP}{Nonlinear Programming}
\newacronym{NLMPC}{NL-MPC}{Nonlinear Model Predictive Control}
\newacronym{MPC}{MPC}{Model Predictive Control}
\newacronym{SQP}{SQP}{Sequential Quadratic Programming}
\newacronym{QP}{QP}{Quadratic Programming}
\usepackage{balance}

\usepackage{blox}

\title{\LARGE \bf
Optimal policy design to mitigate epidemics\\ on networks using an SIS model}

\author{Carlo Cenedese, Lorenzo Zino, Michele Cucuzzella, Ming Cao
\thanks{C. Cenedese is with the Department of Information Technology and Electrical Engineering, ETH Zürich, Zurich, Switzerland. L. Zino, M. Cao and M. Cucuzzella  are with the Faculty of Science and Engineering, University of Groningen, Groningen, the Netherlands.  M. Cucuzzella is also with the Department of Electrical, Computer and Biomedical Engineering, University of Pavia, Pavia, Italy.
        E-mails: {\tt\small ccenedese@ethz.ch, \{lorenzo.zino, m.cao\}@rug.nl, michele.cucuzzella@unipv.it}.}%
\thanks{The work by L. Zino,  and M. Cao is partially supported by the European Research Council (ERC-CoG-771687) and the Netherlands Organization for Scientific Research (NWO-vidi-14134). The work by C. Cenedese is supported by SNSF under NCCR Automation.}%
}

\begin{document}

\maketitle

\begin{abstract}
Understanding how to effectively control an epidemic spreading on a network is a problem of paramount importance for the scientific community. The ongoing COVID-19 pandemic has highlighted the need for policies that mitigate the spread, without relying on pharmaceutical interventions, that is, without the medical assurance of the recovery process. These policies typically entail lockdowns and mobility restrictions, having thus nonnegligible socio-economic consequences for the population. In this paper, we focus on the problem of finding the optimum policies that ``flatten the epidemic curve" while limiting the negative consequences for the society, and formulate it as a nonlinear control problem over a finite prediction horizon. We utilize the model predictive control theory to design a strategy to effectively control the disease, balancing safety and normalcy. An explicit formalization of the control scheme is provided for the susceptible--infected--susceptible epidemic model over a network. Its performance and flexibility are demonstrated by means of  numerical simulations.
\end{abstract}

\section{Introduction}\label{sec:intro}

The ongoing COVID-19 pandemic has highlighted the key role played by public health authorities in enacting \gls{NPI} to ``flatten the epidemic curve" when no effective pharmaceutical treatments such as vaccines are available~\cite{Haug2020,prem2020effect}. However, \gls{NPI} typically entail the implementation of harsh measures, including lockdowns and restrictions of personal freedom of movement, which may yield severe socio-psychological and economic consequences~\cite{Bonaccorsi2020economic}. Thus, they should be implemented keeping a reasonable balance between safety and normalcy. To this aim, the development of accurate tools to predict the course of an epidemic and evaluate the impact of different interventions has become a task of paramount importance for the scientific community, aiming at assisting public health authorities in their decisions. 

The mathematical modeling of epidemics has emerged as a valuable framework to perform such a task~\cite{Nowzari2016,Mei2017,Pare2020review,zino:2021:epidemics_survey,ye:cenedese:2020:epidemics}. Relevant examples can be found in the useful insights provided into the ongoing COVID-19 pandemic~\cite{Giordano2020,Casella2020,Calafiore2020,Gatto2020,dellarossa2020,Carli2020,Parino2021,kohler:2020:robust_MPC_COVID,Bin2021,Parino2021vaccine}. Within this framework, network models have become popular as they allow us to capture the relation between human mobility and the spatial spread of a disease. 
Of particular interest is the problem of understanding how to effectively control the spread of an epidemic disease on a network by acting on the nodal dynamics and on the network structure.

In the literature, such problems have often been addressed assuming limited changes in the network structures, that is, by studying how to re-arrange the network structure or distribute antidote in order to increase the population's resistance against a possible epidemic. Important results have been found by using geometric programming~\cite{Preciado2014}; distributed algorithms have been recently proposed to address this problem~\cite{ramirez2018,Mai2018}. However, the ongoing health crisis has highlighted the importance of having control schemes, which take into account the dynamic evolution of the epidemic process, and can thus be updated online, as the outbreak evolves. While a considerable body of literature has been proposed to study dynamical control strategies for vaccination campaigns and antidote distribution using optimal control theory~\cite{Eshghi2016,Grandits2019} and \gls{MPC}~\cite{selley2015,Kohler2018,Watkins2020}, limited results are available in the absence of effective treatments, that is, when the control action has to focus on the contagion mechanisms rather than on the recovery mechanism. Recently, motivated by the ongoing COVID-19 pandemic, some efforts have been devoted to bridging this gap by proposing feedback control interventions and leveraging \gls{MPC}~\cite{dellarossa2020,Carli2020,kohler:2020:robust_MPC_COVID,Parino2021vaccine}.

Here, inspired by these works, we propose an optimal control approach in order to address the problem of mitigating an epidemic process spreading by means of regional policies that entail activity reductions and targeted mobility restrictions. This holistic approach to \gls{NPI} constitutes a key contribution of this paper and differs from other works in the literature that typically focus on optimizing a specific intervention policy (e.g., social distancing in~\cite{MORATO2020417}) only. Specifically, we consider a discrete-time deterministic \gls{SIS} epidemic model on a network in which two types of control actions are included to mitigate the contagion: actions to reduce the social activity in some regions of the network (modeling, for instance, lockdown measures, closures of activities, and curfews), and policies to limit or ban travels between specific nodes (e.g. cities, provinces, regions, countries). Then, an optimal control problem is formulated to find a strategy that mitigates the spread of the disease in the network, limiting the negative consequences of \gls{NPI}. The proposed optimal control problem takes into account important features, such as the balance between safety and normalcy, the need of keeping the epidemics under control, and the increase of  socio-economic costs associated to the implementation of \gls{NPI}, caused by accumulating factors. Different from many control schemes proposed in the literature (see for example  \cite{Nowzari2016,zino:2021:epidemics_survey} for an overview), our formalization does not (necessarily) have the objective to eradicate the disease (which may be extremely costly ---and practically unfeasible--- using only \gls{NPI}). In contrast, it allows the controller to set an acceptable prevalence of the disease (which may depend on the hospital capacity and may vary across the nodes), providing thus a framework that might realistically be adopted to assist public health authorities in their decision toward mitigating epidemic outbreaks. Finally, we propose several simulations in which we highlight the performance of the proposed control strategy, the benefits and disadvantages of farsighted policies in comparison to myopic ones, and perform a sensitivity analysis on how the objectives of the policy maker affect the final solution. 

The rest of the paper is organized as follows. In Section~\ref{sec:problem_formulation}, we present the SIS epidemic model. In Section~\ref{sec:policies}, we discuss the control policies. In Section~\ref{sec:opt_policy}, we propose our optimal control policy. In Section~\ref{sec:simulations}, we illustrate the results of our numerical simulations. Section~\ref{sec:conclusions} concludes the paper and outlines the avenues of future research.

\subsection{Notation}

We gather here the notation used throughout the paper. The set of nonnegative integer numbers is denoted by $\N$, and the set of nonnegative real numbers is denoted by $\mathbb{R}_+$. Given a vector $\bs x$, we denote by $\bs x^\top$ its transpose and, for a positive definite matrix $S\succ 0$, we denote by $\lVert \bs x\rVert_{S}:=\sqrt{\bs x^\top S \bs x}$ the norm of $\bs x$ weighted by $S$. With $\0$ and $\1$, we denote the all-$0$ and all-$1$ vectors, respectively. Given a matrix $S$, we denote by $\sigma(S)$ its maximum singular value. $S_{ij}$  and $S_i$ denote the element in the $i$-th row and $j$-th column and the $i$-th row of $S$, respectively. 
We use the following notation $\col((x_i)_{i\in[1,\dots,N]}):=[x_1,\dots,x_N]^\top$. 

\section{\gls{SIS} epidemics model over networks}
\label{sec:problem_formulation}
In this section, we introduce the \gls{SIS} epidemic model, illustrate its main properties, and  discuss the assumptions needed to formulate the control problem. 

We consider a discrete-time deterministic \gls{SIS} epidemic model on a network~\cite{Pare2020review}. The network structure arises from the local interactions among several nodes (communities), which represent geographical entities, such as countries, regions, or even cities, thanks to the high flexibility of the model. We assume there are $N$ communities interconnected by $P$ links, which can represent for instance roads, 
air routes or even simple geographic adjacencies. In general, a link between two communities means that there is a flow of people between them. The network is hereafter formalized via an undirected and connected graph $\mc G$, where the communities correspond to the set $\mc V\coloneqq \{1,\dots ,N \}$ of the nodes. A connection between two communities $i,j\in\mc V$ is denoted via an edge $e_\ell$ connecting the two corresponding nodes,  defined as the unordered couple $e_\ell\coloneqq\{i,j\}$. We assume that all the self-loops $\{i,i\}$, $i\in \mc V$ are present. The edge set is the collection of all the edges of the graph, i.e., $\mc E\coloneqq \{e_1,\dots ,e_P\}$. Therefore, the graph is defined as $\mc G \coloneqq (\mc V,\mc E)$. 
The fraction of infected individuals in community $i\in\mc V$ at time $t\in\mathbb{N}$ is denoted by $x_i(t)\in[0,1]$. This quantity describes the temporal evolution of the health  state of community $i$. We assume that the health state evolves according to the dynamics of a discrete-time \gls{SIS} model~\cite{Pare2020review}, i.e., for all time-steps $t\in\mathbb{N}$, we have
\smallskip
\begin{equation}
\label{eq:SIS_dynamics}
x_i^+ = (1-\mu) x_i + (1-x_i)  \bar \beta_i \sum_{j\in \mc V}  \bar A_{ij}x_j\:,
\end{equation}
where $x_i^+:=x_i(t+1)$ and $x_i:=x_i(t)$.    


Next, we discuss the role and physical interpretation of each parameter appearing in \eqref{eq:SIS_dynamics}.
\subsubsection{Recovery rate $\mu\in[0,1]$} It is the rate at which the individuals manage to recover from the disease. Here, we assume that the recovery rate is constant across the entire population and cannot be increased by the controller, capturing those epidemics for which a cure is not yet developed, e.g., the early spread of COVID-19.
\begin{assumption}[Positive recovery rate]
\label{ass:positive_recovery_rate}
For all $i\in\mc V$, the recovery rate $\mu>0$ is constant and strictly positive.\hfill\QEDopen
\end{assumption}\smallskip

\subsubsection{Infection rate $\overline{\beta}_i>0$} This represents the rate at which the individuals in community $i\in\mc V$ become infected when they get in touch with others. The higher this value is, the easier people become infected. Such a value can differ among the communities, e.g., due to the implementation of different NPIs. Infection rates are gathered in the $n$-dimensional vector $\bar\beta$. 
The trajectory of $\bs{x}(t)\coloneqq \col((x_i(t))_{i\in\mc V})$ can naturally have two possible behaviors, depending on the model parameters. Either it converges to the disease-free equilibrium, i.e., $\bs{x}(\infty)=\0$, or the disease-free equilibrium becomes unstable, and the trajectory converges to a (unique) endemic equilibrium with  $x_i(\infty)>0$, for all $i\in\mathcal V$~\cite{Pare2020,Liu2020}. The threshold between these two regimes depends on whether $\sigma(\diag(\bar\beta)\bar A)/\mu$ is smaller or greater than $1$. A simpler---network-independent--- sufficient (but nonnecessary) condition for the trajectory to converge to an endemic equilibrium can be established by requiring that $\overline \beta_i/\mu > 1$ for all $i\in \mc V$. In this work, we focus on the case in which the disease does not  die out naturally, and thus \gls{NPI} have to be put in place toward mitigating the epidemic outbreak. Hence, we make the following assumption. 
\begin{assumption}[Disease spreading]
\label{ass:disease spreading}
For all $i\in\mc V$, the infection rate satisfies $\overline \beta_i > \mu $.\hfill\QEDopen
\end{assumption}

\subsubsection{Communities interaction $\overline A_{i,j}\in[0,1]$} People in $j\in\mc V$ can move to the neighboring communities, denoted by $\mc N_j\coloneqq\{\ell\in\mc V\,:\, \{j,\ell\}\in\mc E \}$, and interact with people there. The population that flows from $j$ to $i\in\mc N_j$ influences the health state of community $i$, i.e., $x_i$. This is modeled in \eqref{eq:SIS_dynamics} via the weighted adjacency matrix $\overline A\coloneqq [\overline A_{ij}]_{i,j\in\mc V}\in\mathbb{R}_+^{N\times N}$, where $\overline{A}_{ij}$ is the weight associated to the edge $(i,j)$ of $\mc G$.
The diagonals of $\overline A$ represent the part of population that remains in the same community. We now  introduce the following assumption on this matrix.
\smallskip
\begin{assumption}[Stochastic and positive diagonals]
\label{ass:doubly_stoch_A_bar}
The weighted adjacency matrix $\overline A$ associated with the undirected graph $\mc G$ is  stochastic, i.e., $\overline A\mathbf{1}=\mathbf{1}$, and with strictly positive diagonal entries, i.e., $\overline A_{ii}>0$, for all $i\in\mathcal V$.~\hfill\QEDopen
\end{assumption} 
\smallskip
In the following section, we discuss the control policies.

\section{Control policies}
\label{sec:policies}
To mitigate the spread of a disease, a policy maker can apply \textit{endogenous} or \textit{exogenous} \gls{NPI} to each local community.  The former category includes all those measures put in place inside the local community $i$, e.g., lock-downs, usage of face masks or encouraging social distancing.
The latter instead concerns those measures that limit the inflow of people from neighboring communities, i.e., implementation of travel bans or requesting quarantine upon entrance. 
It is clear that these policies directly affect the dynamics in \eqref{eq:SIS_dynamics}. Specifically,  for each community $i$,  the  endogenous measures reduce the infection rate $\overline\beta_i$, while  the exogenous ones act on $\overline A$. The static values of $\overline\beta_i$ and $\overline A$ refer to the \textit{uncontrolled} evolution of the epidemic, while the intervention of the policy makers transforms these parameters into time-varying functions $\beta_i(t)$ and $A(t)$, which can be seen as inputs to be designed in a (possibly) optimal way to control the epidemic spreading. Thus, the dynamics of the \textit{controlled} evolution of the SIS epidemic become
\smallskip
\begin{equation}
\label{eq:SIS_dynamics_controlled}
x_i^+ = (1-\mu) x_i + (1-x_i) \beta_i(t) \sum_{j=1}^N  A_{ij}(t)x_j\:,
\end{equation}   
as illustrated in the schematic in Fig.~\ref{fig:network_sis}.

 \begin{figure}
  \centering
\begin {tikzpicture}
\tikzstyle{peers}=[draw,circle, text=black,  fill=white,inner sep=0pt, minimum size=.6cm, thick]

\node[peers] (0) at (-3.5,-.75) {$x_i$};
\node[rounded corners,fill=black!10] (1) at (0,0) {\small interactions};
\node[peers] (2) at (1.2,1.2) {$x_h$};
\node[peers] (3) at (1,-1.4) {$x_k$};
\node[peers] (4) at (2,0) {$x_j$};

\node[rounded corners,fill=black!10] (R) at (-1.5,-1.5) {\small recovery};

\path (0) edge[thick,->] node[below] {\small $\mu$} (R) ;
\path (0) edge[thick,<-] node[above] {\small $\beta_i(t)$} (1) ;
\path (1) edge[thick,<-] node[above] {\small $A_{ij}(t)$} (4) ;
\path (1) edge[thick, <-] node[left] {\small $A_{ih}(t)$}(2);
\path (1) edge[thick, <-] node[left] {\small $A_{ik}(t)$}(3);

\end{tikzpicture}
\caption{Schematic of the controlled network SIS model. The time-varying weighted adjacency matrix $A(t)$ determines the interaction that community $i$ has with other communities. The time-varying infection rate $\beta_i(t)$ determines the new infections, while the constant recovery rate $\mu$ governs the recovery process. }
\label{fig:network_sis}
\end{figure}
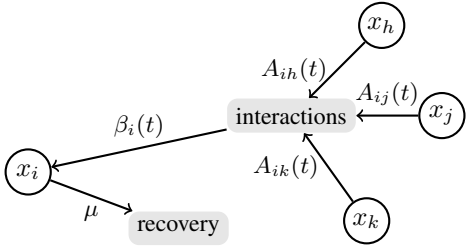

If the policies put in place are effective, the value of the infection rate should decrease, i.e., $\beta_i(t)\in[0, \overline\beta_i]$ for all $t\in\mathbb{N}$. We consider policies that affect the uncontrolled coefficients  in \eqref{eq:SIS_dynamics} linearly. So, we define time-varying functions $\beta_i(t)$ and $A(t)$ in \eqref{eq:SIS_dynamics_controlled} as 
\smallskip
\begin{subequations}
\label{eq:policies}
\begin{equation}
\label{eq:policy_beta}
\beta_i(t) \coloneqq \overline{\beta}_i - v_i(t) \,,\: \forall i \in\mc V \,,\\
\end{equation}
\begin{equation}\label{eq:policy_A}
A_{ij}(t) \coloneqq \overline{A}_{ij}-W_{ij}(t)\,,\: \forall i,j \in\mc V\,,
\end{equation}
\end{subequations}
where $v_i(t)\in[0,\ \overline{\beta}_i]$. Furthermore, it is reasonable to assume that the policy cannot implement new connections among the community, this translates into the following set of constraints on  $W(t)\coloneqq[W_{ij}(t)]\in \mathbb{R}^{N\times N}$
\begin{subequations}
\label{eq:cnd_W}
\begin{equation}
\label{eq:cnd_1_W}
W(t)\mathbf{1} = \0 \,,
\end{equation}
\begin{equation}
\label{eq:cnd_2_W}
\overline{A}_{ij}-W_{ij}(t) \geq 0\,,
\end{equation}
\begin{equation}
\label{eq:cnd_3_W}
\overline{A}_{ij}=0\implies W_{ij}(t)= 0\,.
\end{equation}
\end{subequations}
 
%

 
From \eqref{eq:policy_A} and \eqref{eq:cnd_W}, it follows that the time-varying weighted adjacency matrix $A(t)$ defines a directed subgraph $\mc G(t)\subseteq\mc G$, which is composed of the same nodes and a subset of the edges of $\mc G$, i.e., $\mc G(t)\coloneqq (\mc V, \mc E(t))$ with $\mc E(t)\subseteq \mc E$, so it may not be connected at every $t\in\N$. 

\section{Optimal control policy}
\label{sec:opt_policy}

In this section, we present the main result of this paper on optimal control policies. 

\subsection{Problem formulation}
The adoption of stringent policies carries costs including both monetary (e.g. recession) and social (e.g. personal restrictions) aspects~\cite{Bonaccorsi2020economic}. In this section, we formalize the problem of choosing the optimal \gls{NPI} that minimize the cost while keeping under control the epidemic. Therefore, at each time instant $t\in\N$, we have to design the values of the control actions $W(t)$ and $v_i(t)$ for all time instants $k$ in the prediction horizon $\mc T \coloneqq\{t,t+1,\dots,t+T_{\tup h}\}$, the length of which ($T_h\in\N$) may vary due to the specific epidemic. In the remainder, we use the index $k$ when we refer to an instant belonging to $\mc T$ and $t$ otherwise. 
 
Each population $i\in\mc V$ is assumed to establish a desired trajectory $\hat x_i(k)$, for all $k\in\mc T\setminus \{t\}$, that is   the fraction of infected population $x_i(k)$  considered to be acceptable at time $k$. An interesting case is the one in which the value of $\hat x_i(k)$ is greater than $0$ at the beginning of $\mc T$, and it decreases over time, i.e., $\bs{\hat x}(k)\geq  \bs{\hat x}(k+1)$ for all $k,k+1\in\mc T\setminus \{t\}$. Its slope suggests how aggressive the desired community's policies should be. Notice that the terminal value $\hat x_i(t+T_{\tup h}+1)$ is not necessarily constrained to be $0$. In fact, it is reasonable that some communities accept a small fraction of infected in exchange for relaxed \gls{NPI}.  

Next, let us define the vector of all policies put in place at time instant $t$ by agent $i$ as the $(N+1)$-dimensional vector 
$\bs{u}_i(t)\coloneqq \col(v_i(t), W_{i1}(t),\dots ,W_{iN}(t))$.
The cost for each community $i$ consists of two antagonizing components. The first is the \emph{health-care cost} $c^{\tup{HC}}_i(x_i(t);\hat x_i(t))$ due to the presence of more infected than the desired ones $\hat x_i(t)$, which is individually decided by each community and communicated to the health authority. If the optimal policies leads to $ x_i(t)\leq \hat x_i(t)$, then the cost is $0$; otherwise it is assumed quadratic in the difference with respect to the desired value, and thus it reads as
$$c^{\tup{HC}}_i(x_i(t);\hat x_i(t))\coloneqq q_i(t)\left(\max\left\{0,x_i(t)- \hat x_i(t)\right\}\right)^2, $$
with $q_i(t)>0$ being a (possibly time-varying) weight.  On the other hand, we consider a (quadratic) \emph{control cost} associated with the implementation of control policies. So, the global cost that the $N$ communities face over the prediction horizon reads as
\smallskip
\begin{equation*}
\label{eq:cost_global}
J(\bs x,  \bs u ) \coloneqq \sum_{k\in\mc T } \sum_{i\in\mc V}  c^{\tup{HC}}_i(x_i(k+1); \hat x_i(k+1)) +  \lVert \bs{u}_i(k)\rVert^2_{S_i(t)}\,,
\end{equation*}
where  $S_i(k)\succ 0$ is the diagonal matrix of the $N+1$  weights associated to the \gls{NPI} $\bs{u}_i(k)$. The fraction of infected and the  adopted policies of all the populations at $t$ are $\bs{x}(t)$ and $\bs{u}(t)\coloneqq\col((\bs{u}_i(t))_{i\in\mc V})$, respectively. The pair $\bs x=\col((\bs{x}(t+1))_{t\in\mc T})$ and $  \bs u=\col((u(t))_{t\in\mc T})$ denotes in compact form  all the variables involved. In our general formulation, the weights associated to the health-care cost (i.e., $q_i(t)$) and to the control cost (i.e., $S_i(t)$) are time-varying, since the same fraction of infected individuals or the same level of \gls{NPI} may yield different costs depending on the timing, for example due to the increasing hospital preparedness or the accumulation of socio-economic costs. It is worth noticing that only $2P-2N$ values of $W(k)$ can be freely designed at each time instant $k\in\mc T$. In fact, for each one of the $P-N$ off-diagonal edges, we can set $W_{ij}(k)$ and $W_{ji}(k)$; then the weight associated with the self-loop is constrained by \eqref{eq:cnd_1_W}. Thus, the number of elements of $u(t)$ to be optimized is $2P-N$.

Finally, we cast the problem of designing the  optimal policies $\bs u^*$ for the control dynamics \eqref{eq:SIS_dynamics_controlled} 
\begin{equation}\begin{array}{lll}
\label{eq:opt_problem}\tag{$\mc P$}
(\bs{x^*},\bs{u^*})&=&\argmin  J(\bs x,  \bs u )\,,\:\\&& \text{s.t. } (\bs x,  \bs u )\in\Omega\,,
\end{array}\end{equation} 
where
\begin{equation*}
\Omega \coloneqq \left\{(\bs x,\bs u) |\, v_i(t)\in[0,\overline \beta_i],\, \eqref{eq:SIS_dynamics_controlled} ,\,\eqref{eq:policies},\,
\eqref{eq:cnd_W} \text{ hold } \forall k\in\mc T \right\}. 
\end{equation*}
This problem belongs to the class of \gls{NLMPC}, due to the highly nonlinear and nonconvex controlled \gls{SIS} dynamics in \eqref{eq:SIS_dynamics_controlled}. It is well known that these problems are hard to solve in their original form. Nevertheless, we establish  that a  feasible solution to \eqref{eq:opt_problem} always exists, and thus  that the problem is worth to be studied.
\smallskip
\begin{proposition}[Solution existence]
\label{prop:solution_existence}
For every initial condition $x(t)\in [0,1]^N$, there exists at least one pair $(\bs{x^*},\bs{u^*})$ that is a solution to \eqref{eq:opt_problem}.\hfill \QEDopen
\end{proposition}
\begin{proof}
From the constraints' definition \eqref{eq:SIS_dynamics_controlled} and 
\eqref{eq:cnd_W}, it follows  that the pair $\bs u=\0$ and $\bs x$, obtained via \eqref{eq:SIS_dynamics}, always lies in $\Omega$, hence $\Omega \neq \varnothing$. The box constraints on the control inputs  and the fact that the dynamics \eqref{eq:SIS_dynamics_controlled} are positively invariant with respect to $[0,1]^N$, see~\cite[Lem.~1]{Pare2020}, implying that $\Omega$ is bounded. For every constraint, the set of points satisfying it is closed. Therefore, $\Omega$ is the intersection of closed sets, and  hence it is closed itself. Since $J$ is continuous, we  invoke the Weierstrass theorem to conclude the existence of $(\bs{x^*},\bs{u^*})\in\Omega$ that minimizes $J$, and, in turns, solves \eqref{eq:opt_problem}. 
\end{proof}
\smallskip

\subsection{\gls{NLMPC} solution algorithm}
In the literature, there are several approaches to solve  \gls{NLMPC}. Specifically, in order to solve nonlinear constrained optimization problems with differentiable cost functions, the most popular approaches are based on the \gls{SQP}. As shown in~\cite[Ch.~18]{nocedal2006sequential},~\cite{schittkowski:1986:NLPQL} and reference there in, these methods provide excellent convergence properties and ensure a fast convergence to a (local) optimum of the original nonlinear problem. The \gls{SQP} is an iterative algorithm in which, during each iteration $p$, a candidate optimal trajectory $(\bs{x}^p,\bs{u}^p)$ is computed as the solution of a \gls{QP}, obtained by linearizing the constraints and approximating the cost via a quadratic funciton, see~\cite[Alg.~18.3]{nocedal2006sequential}. The linearization is performed with respect to the trajectory $(\tilde{\bs{x}}^{p-1},\tilde{\bs{u}}^{p-1})$ defined as
$$(\tilde{\bs{x}}^{p-1},\tilde{\bs{u}}^{p-1})=(\tilde{\bs{x}}^{p-2},\tilde{\bs{u}}^{p-2}) + \alpha^{p-1} ({\bs{x}}^{p-1},{\bs{u}}^{p-1}),$$
where $({\bs{x}}^{p-1},{\bs{u}}^{p-1})$ is the solution of the \gls{QP} solved at the previous iteration, and $\alpha^{p-1}$ can be computed for example via line search as in~\cite[Eq.~18.28]{nocedal2006sequential}.
If the candidate solution satisfies some convergence conditions then $(\bs{x^*},\bs{u^*})=(\bs{x}^p,\bs{u}^p)$; otherwise a new iteration is performed.

The sole nonlinearity in the constraints of \eqref{eq:opt_problem} is associated with the dynamics in \eqref{eq:SIS_dynamics_controlled}. The effectiveness of the \gls{NPI} depends on the predictive accuracy of the linearized dynamics. This difference is studied via numerical simulations in Section~\ref{sec:simulations}, where it is shown that the linearized dynamics allow to obtain a very high prediction accuracy for a sufficiently long prediction horizon. Thus, the used \gls{NLMPC} algorithm generates an effective solution for the control problem.
  
Finally, in Figure~\ref{fig:control_scheme} we depict the complete control scheme to solve the problem of designing optimal \gls{NPI} to control an \gls{SIS} type dynamics. Specifically, at each time instant $t$ an instance of \eqref{eq:opt_problem} is cast. The optimal trajectory  $(\bs{x^*},\bs{u^*})$ is computed over $\mc T$ via \gls{SQP}. Finally, a receding horizon approach is implemented by applying only \gls{NPI} associated to the first instant, i.e., $u^*(t)$, and then the loop starts again. This implementation allows us to minimise the prediction error inherently present in the prediction of the model.


\begin{figure}
\centering
\begin{tikzpicture}[scale=1, transform shape]
\bXInput{A}

\bXStyleBloc{rounded corners,fill=black!10,text=blue}
\bXBloc[3]{controller}{\eqref{eq:opt_problem}}{A}
\bXLink[$\bs{\hat x}$]{A}{controller}
\begin{scriptsize}
        \bXLinkName[2.8]{controller}{NPIs design}
\end{scriptsize}
\bXStyleBloc{rounded corners,fill=blue!10,text=blue}
\bXBloc[5]{plant}{\eqref{eq:SIS_dynamics_controlled}}{controller}
\bXLink[$u^\ast(t)$]{controller}{plant}
\begin{scriptsize}
        \bXLinkName[2.8]{plant}{SIS dynamics}
\end{scriptsize}

\bXOutput[3]{Z}{plant}
\bXLink[$x(t)$]{plant}{Z}
\bXReturn{plant-Z}{controller}{}
\end{tikzpicture}
\caption{Control scheme  for the design and implementation via receding horizon of optimal \gls{NPI} to flatten the pandemic curve.}
\label{fig:control_scheme}
\end{figure}
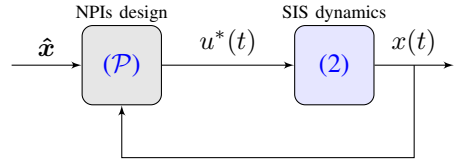


\section{Simulations}\label{sec:simulations}
\begin{figure*}[t]
    \includegraphics[width=\columnwidth]{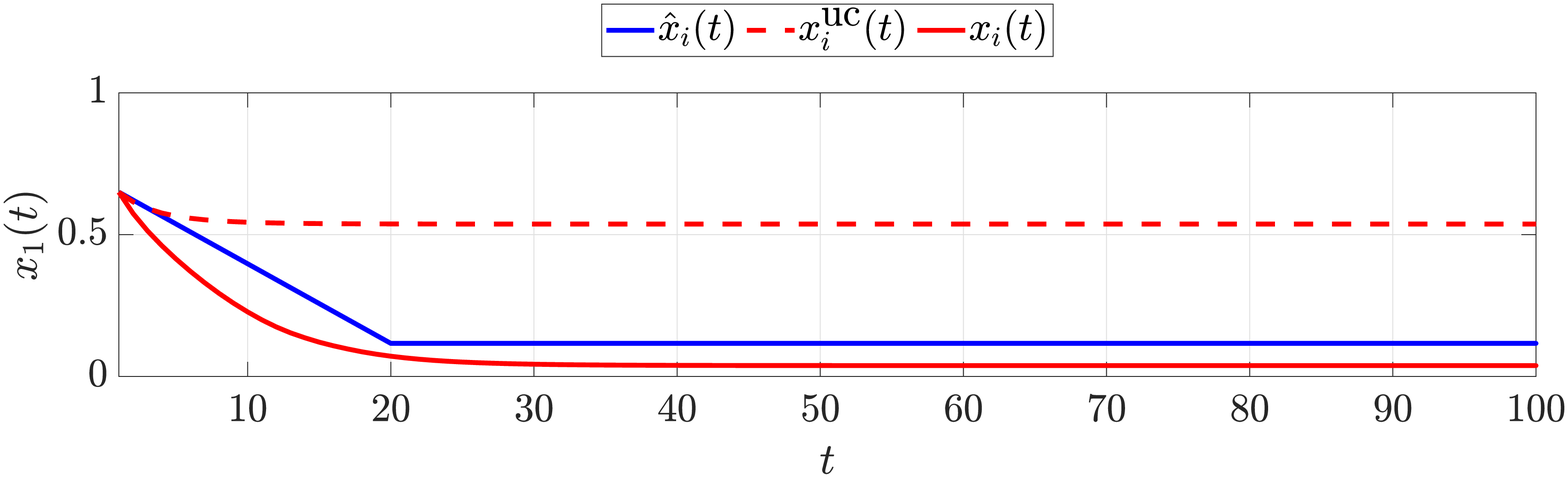}\hfill
    \includegraphics[width=\columnwidth]{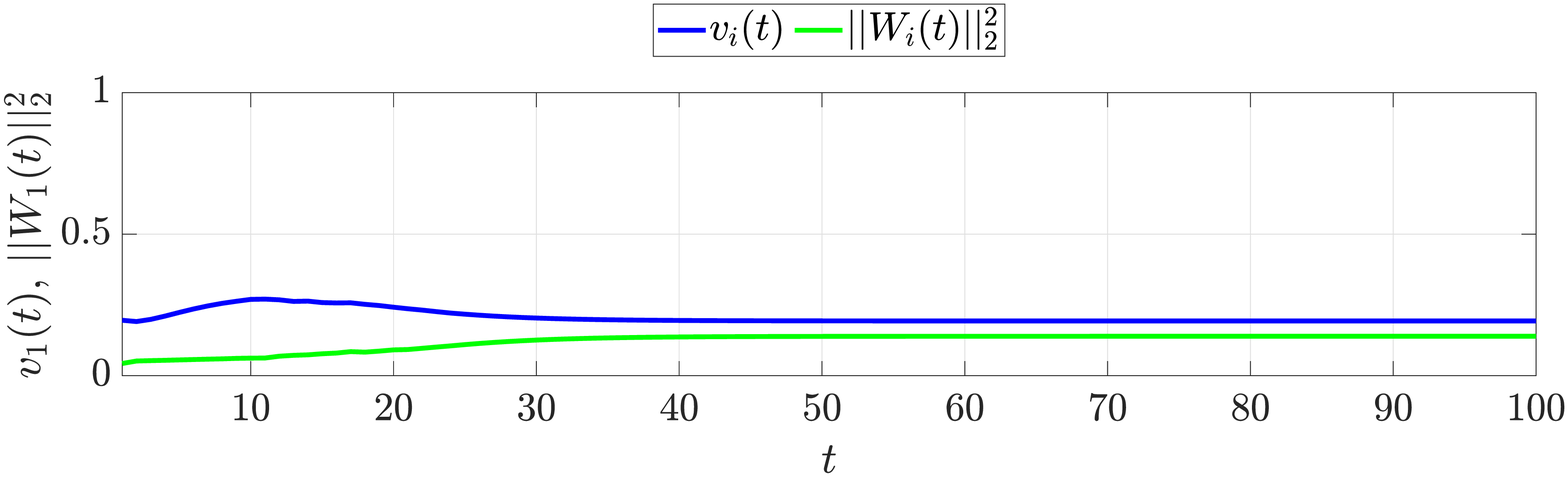}
    \\[\smallskipamount]
    \includegraphics[width=\columnwidth]{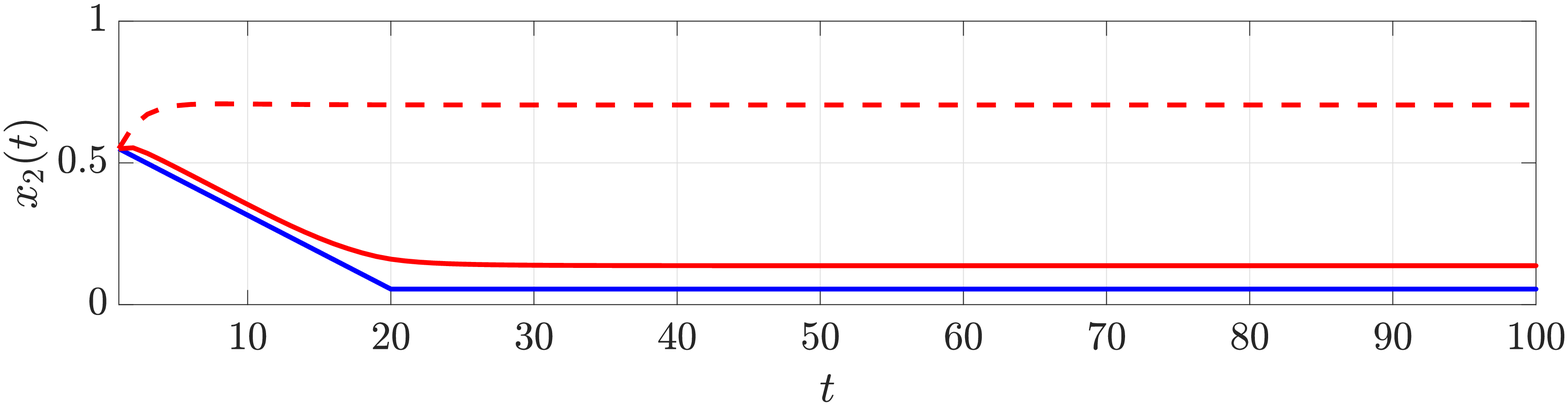} \hfill
    \includegraphics[width=\columnwidth]{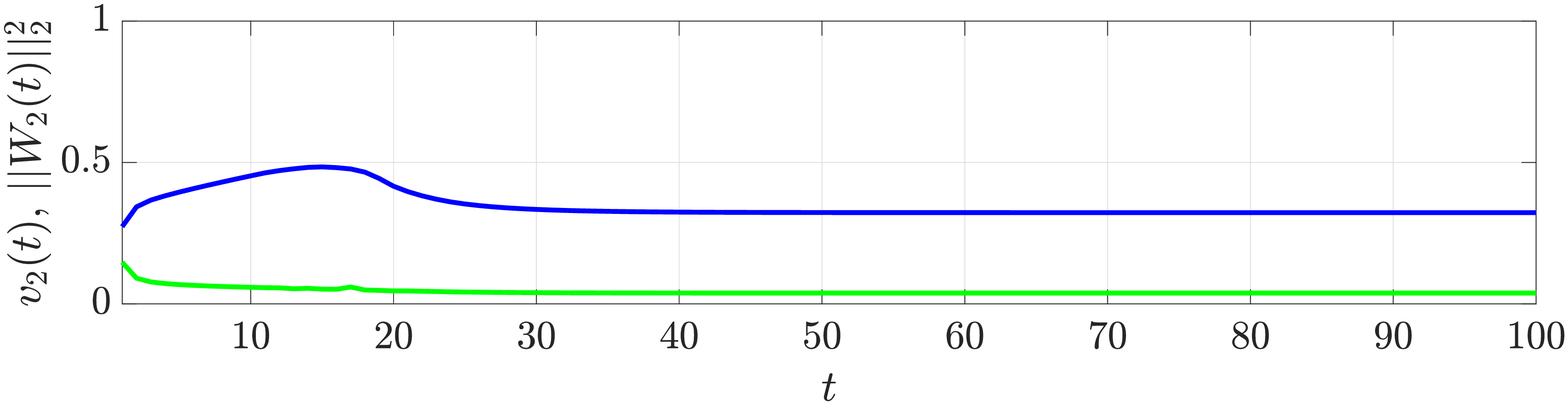}
    \\[\smallskipamount]
    \includegraphics[width=\columnwidth]{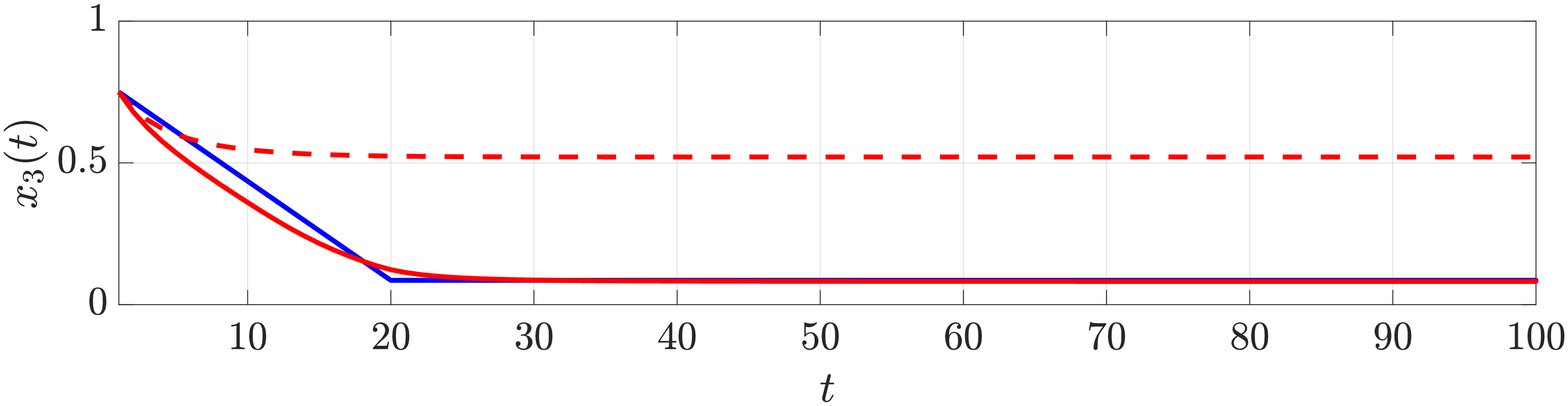}\hfill
    \includegraphics[width=\columnwidth]{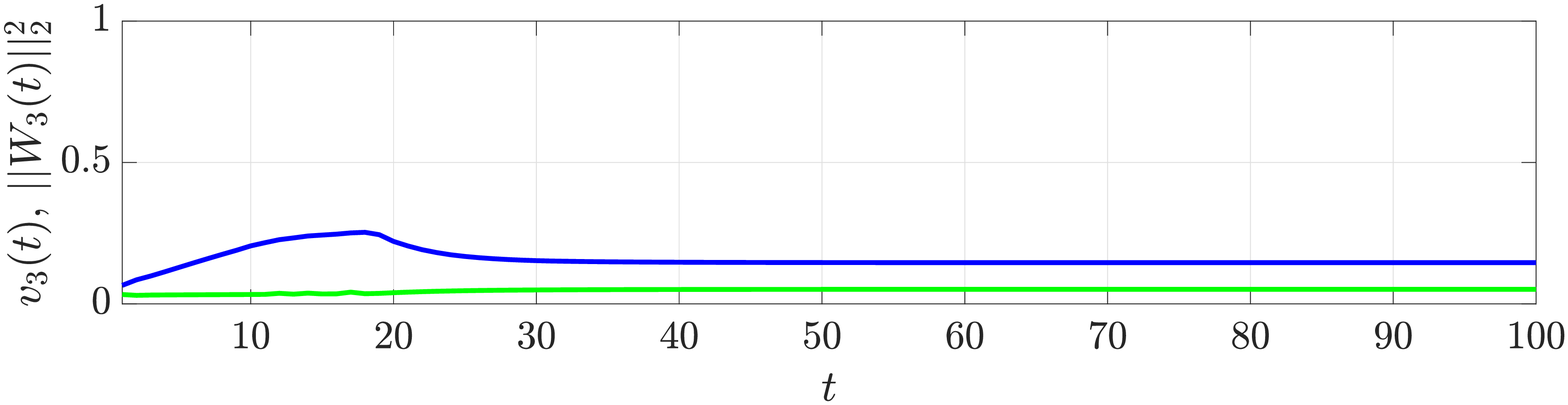}
    \\[\smallskipamount]
    \subfloat[]{\includegraphics[width=\columnwidth]{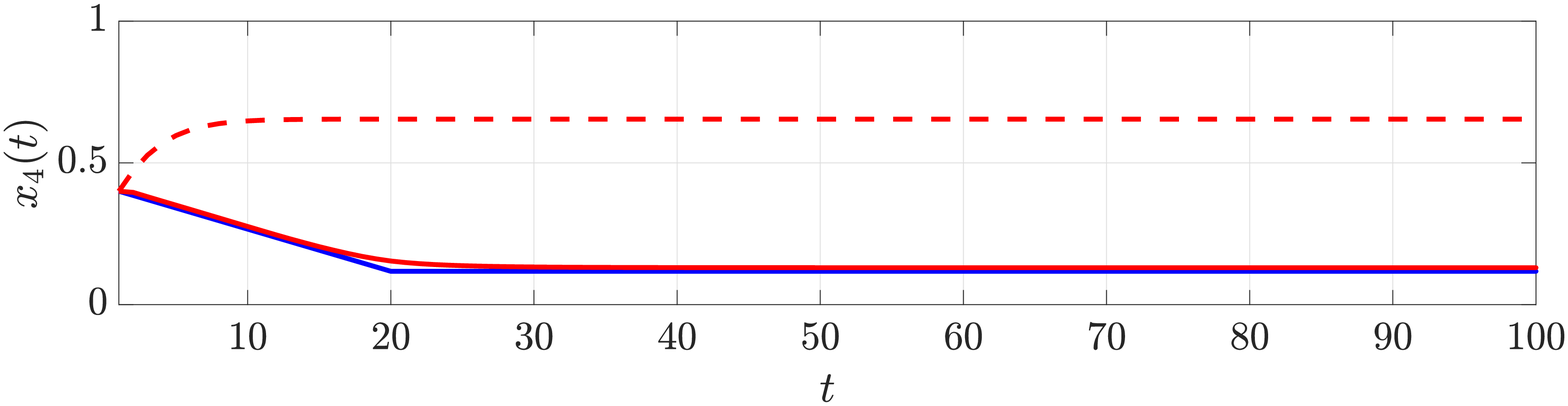}\label{fig:trajectory}}\hfill
    \subfloat[]{\includegraphics[width=\columnwidth]{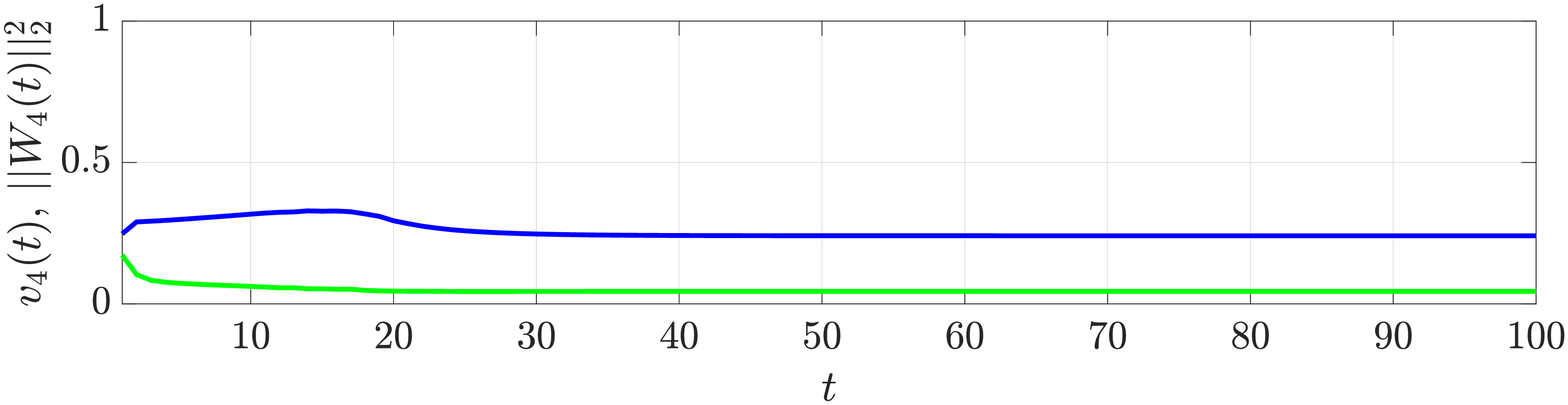}\label{fig:control}}
    \caption{In (a), we show the trajectories of the uncontrolled dynamics  $x^{\tup{uc}}_i(t)$ (red dashed curve), of the reference $\hat x_i(t)$ (blue solid curve), and of the controlled dynamics  $x_i(t)$ obtained via the control scheme in Figure~\ref{fig:control_scheme} (red solid curve) for all $i\in\mc V$. In (b), we plot the amount of endogenous (blue) and exogenous (green) control applied to every community $i\in\mc V$ to obtain $x_i(t)$, respectively $\beta_i(t)$ and $\lVert W_i(t) \rVert^2_2$. }\label{fig:trajectory_and_control}
\end{figure*}

In this section, we present several numerical simulations that validate  the procedure proposed to design optimal \gls{NPI} and provide insightful information on how different parameter choices (e.g., the weights on the control and the length of the prediction horizon) influence the final optimal policies.

\subsection{Reference tracking and reproduction number}
The simulations are performed over a randomly generated connected network with $N=4$  that is described by a weighted adjacency matrix with off-diagonal nonzero entries chosen from a normal distribution with mean $0.1$ and variance of $0.2$, and diagonal entries such that $\bar A\1=\1$, obtaining
\begin{equation}
\label{eq:A_bar}
\overline{A} = \begin{bmatrix}
0.7 &   0.17 &         0 &   0.13\\
    0.42 &    0.31 &    0.16 &    0.11\\
         0 &    0.12 &    0.88 &         0\\
    0.28 &    0.1 &         0 &   0.62
\end{bmatrix}.
\end{equation}
The uncontrolled infection rate $\overline \beta_i$ of each node is selected uniformly and randomly in the interval $[0.3,0.6]$, obtaining $\overline \beta = [    0.3,\, 0.59, \,0.3\,, 0.45 ]$.
The recovery rate is chosen to be $\mu=0.15$. So, the  uncontrolled dynamics  converge over time to an endemic equilibrium~\cite{Pare2020}.  Note that, since $\mu$ is the inverse of the time for recovering, if one time-step corresponds to a day, then this setup may represent a disease like seasonal flu, for which (approximately) $6$ days are  enough to recover. On the other hand, if one time-step is interpreted as one week, the setup may capture diseases with longer recovery time, e.g., gonorrhea. The initial fraction of infected for each population is set as $\bs{x}(0)=[0.65,\, 0.55,\, 0.75,\, 0.40]$, that is, a scenario of an endemic disease. We denote the trajectory of the fraction of infected individuals in the uncontrolled dynamics \eqref{eq:SIS_dynamics} by $\bs{x}^{\tup{uc}}(t)$. The desired reference trajectory $\hat x_i(t)$ of each community is assumed to start close to $x_i(0)$ and linearly decrease until it reaches its terminal value at $t=20$ that is $\bs{\hat x}(20) = [0.1168,\, 0.0548,\, 0.0856,\, 0.1175]$. The length of the prediction horizon is set equal to $T_h=10$ and we will discuss its optimal value in Section~\ref{sec:optimalT}. The weights $q(t)$ and $S(t)$ are chosen constant in time, uniform across the different nodes, and equal to $q_i(t)=1$, for all $i\in\mathcal V$, while $S$ is so that all off-diagonal terms are equal to $0$, the diagonal terms that correspond to entries of $v$ are equal to $0.2$ and those that correspond  to entries of $W$ are equal to $0.05$. Note that the weights are such that applying the same amount of exogenous and endogenous control respectively have similar costs.

In Figure~\ref{fig:trajectory}, we present the trajectory of $x_i(t)$,  obtained from \eqref{eq:SIS_dynamics_controlled}, by following the control scheme in Figure~\ref{fig:control_scheme} (red solid curves), compared with the value of the uncontrolled SIS dynamics (red dashed curves), and the reference (blue solid curve). As expected from having a higher weight on the healthcare cost, the value of $x_i(t)$ remains relatively close to the desired one $\hat x_i(t)$. From Figure~\ref{fig:control}, it is clear that, in this scenario, acting on local restrictions as lock-down and social distancing is more effective than implementing travel bans. This is consistent with empirical observations during the ongoing COVID-19 pandemic, suggesting that travel bans are more effective in the early stages of an epidemic outbreak, i.e., when $\bs x(0)$ is close it $\0$~\cite{Parino2021}. In the first instances, the control action is more intense due to the  decreasing reference and, after $t=20$, it is relaxed. This can also be assessed via the ratio $\beta_i/\mu$ depicted in Figure~\ref{fig:R_t}. In fact, during the initial time instants, all this ratios decrease  reaching values lower than $1$.  This suggests that, by adopting the \gls{NPI} implemented in the first phase for the entire duration of the simulation, the trajectory would converge to the disease-free equilibrium, at the cost of dramatically increasing the social and economical impact of the control. 

\begin{figure}[th]
\centering    
    \subfloat[]{ \includegraphics[width=0.45\columnwidth]{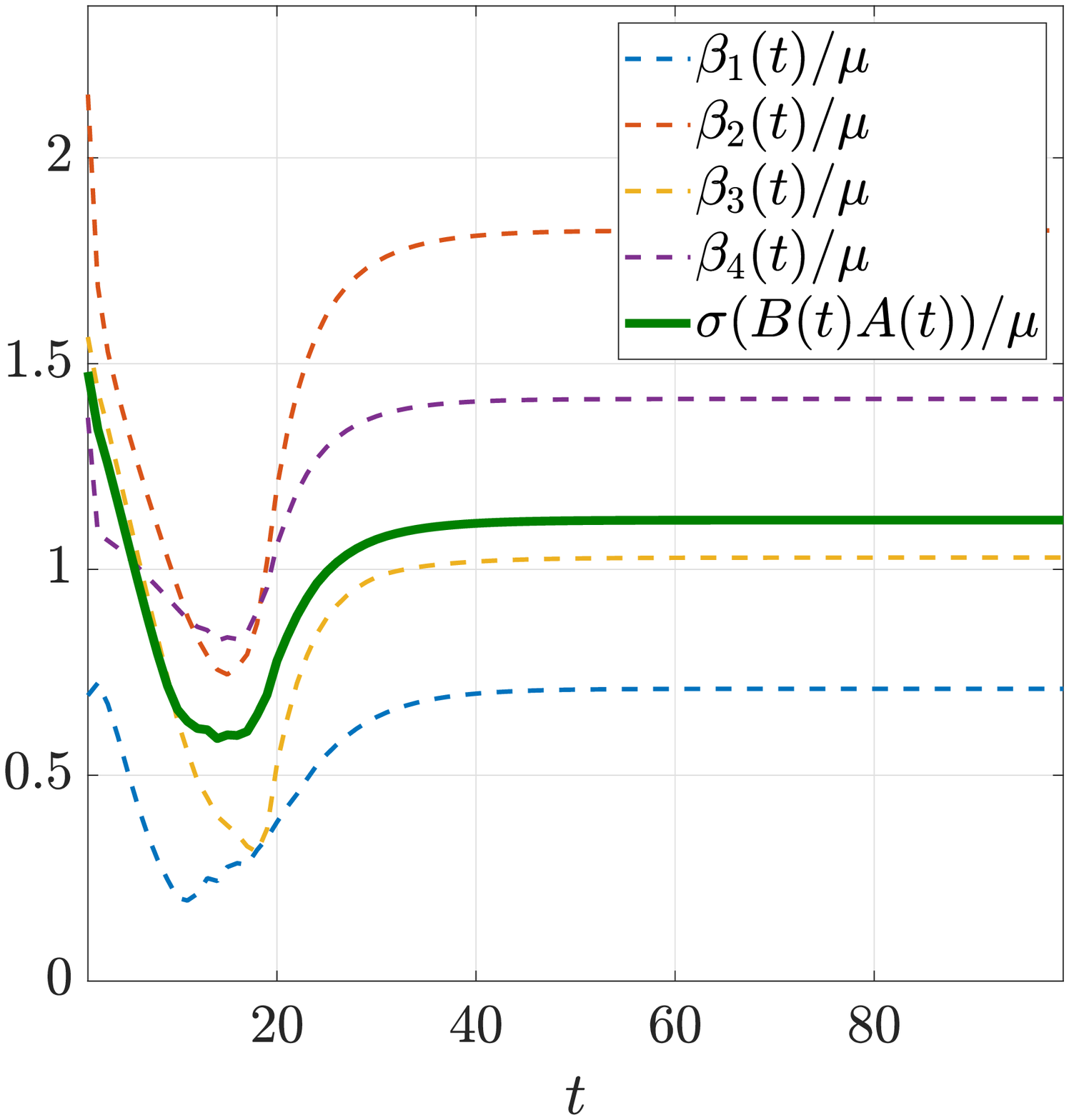}\label{fig:R_t}} 
 \subfloat[]{ \includegraphics[width=0.47\columnwidth]{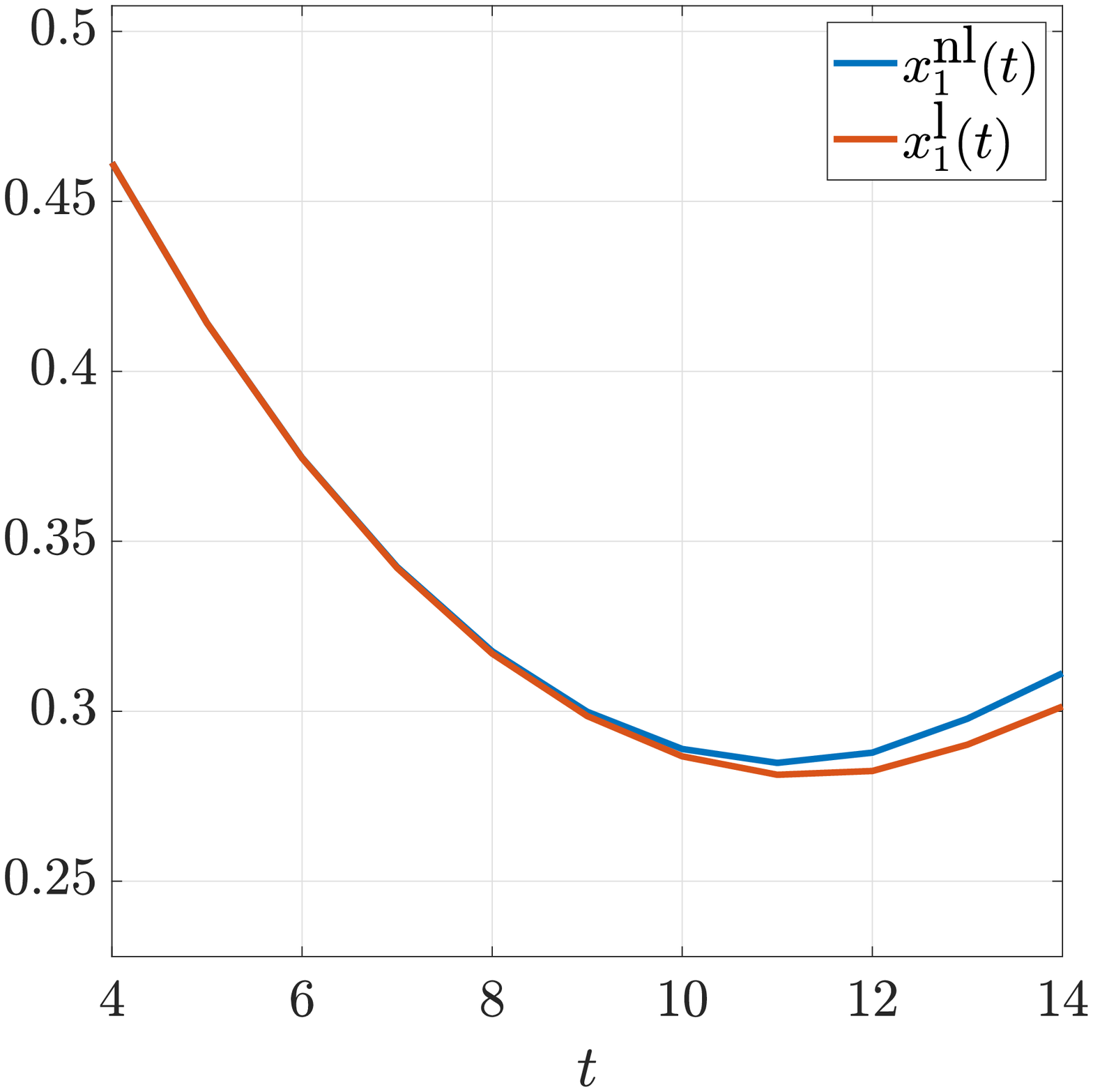}\label{fig:nl_l}}     \caption{In (a), we show the evolution of the ratio $\beta_i(t)/\mu$ for all the communities and the  value of  $\sigma(B(t)A(t))/\mu$. In (b), we compare the trajectory predicted by the linearized SIS model, i.e., $x^{\tup{l}}(t)$ (red) and the one obtained from the open loop controlled nonlinear dynamics $x^{\tup{nl}}(t)$ (blue).}\label{fig:R_t_and_nl_l}
\end{figure}

In Figure~\ref{fig:R_t}, we also report the evolution of the singular value $\sigma(B(t)A(t))/\mu$ where $B(t)\coloneqq \diag(\beta(t))$, which is known to determine the behavior of the epidemic process~\cite{Mei2017,Pare2020}.  It is also strictly related to the effective reproduction number $R_t$, which can be obtained by multiplying such a quantity by the fraction of susceptible individuals in the population. As expected, the value at which  $\sigma(B(t)A(t))/\mu$  settles is above $1$, since the terminal value of the reference $\bs{\hat x}(t)$ is greater than the disease-free equilibrium.  It is interesting to notice that the optimal control $W(t)$  tends to isolate the communities by increasing the weight of each self-loop and decreasing the off-diagonal entries. Analyzing more complex graphs may lead to the emergence of different patterns in the exogenous control actions, and this is left for future research.

In Figure~\ref{fig:nl_l}, the accuracy of the linearized model used to solve \eqref{eq:opt_problem} at time $t=4$ is shown. By solving the optimization problem, we obtain the trajectories $\bs{x}^*(t)$ and $\bs{u}^*(t)$ for $t\in\mc T=\{4,\dots,14\}$, which we denote by $\bs{x}^{\tup{l}}(t)$. We compute $\bs{x}^{\tup{nl}}(t)$ by applying the whole $\bs{u}^*$ to \eqref{eq:SIS_dynamics_controlled} and letting it evolve in an open loop. The estimation of the nonlinear behavior of the dynamics is almost perfect during the first six time instants of $\mc T$, while it becomes less accurate during the second half. An even longer prediction horizon would exacerbate this difference affecting the performance, as we show in the next section. 

Overall, the policies computed by the proposed control scheme generate noteworthy effects in controlling the epidemics by achieving the two main objectives of the proposed approach, i.e., an acceptable level of infected and at a low social and economical price.
Nevertheless, it is evident that the optimal solution to the problem \eqref{eq:opt_problem} highly depends on the choice of the weights $q_i(t)$ and $S_i(t)$ in the cost function $J(\bs x,\bs u)$. To better understand the effect of the endogenous and exogenous measures on the (steady-state) health state of the overall population, we have performed a sensitivity analysis of the weights applied to the control action.  For all $i$ and $t$, we set $q_i(t)=1$ and $ S_i(t)$ equal to the block diagonal matrix that has on the diagonal $s_v$ and $\frac{s_w}{N} I$, which are the weights on $v_i(t)$ and $W_i(t)$, respectively. In order to compare the performance obtained with different weights, we introduce the following index:
\begin{equation}\label{eq:pi}
\pi := \frac{||{\bs x}(\infty)-\bs{\hat x}(\infty)||}{||{\bs x}(\infty)-{\bs x}^{\tup{uc}}(\infty)||}\,.
\end{equation}
As shown in Figure \ref{fig:weights}, we obtain larger values of $\pi$ (i.e., the system converges close to the endemic equilibrium $\bs{x}^{\tup{uc}}$) when the weights of the control actions are large, while  smaller values of $\pi$ (i.e., the system converges close to the desired reference $\bs{\hat x}$) arise  if $s_v$ and $s_w$ are small. 
Furthermore, from Figure \ref{fig:weights} one can observe that changing the weight $s_v$ makes the index $\pi$ change more significantly than changing the weight of the exogenous actions. Then, we can conclude that the endogenous actions are more effective than the exogenous ones. However, we can also observe that when the cost of implementing for instance lock-downs is very high (i.e., larger values of $s_v$), putting in place travel bans is very beneficial.
\begin{figure}[t]
\centering     
\includegraphics[width=1.1\columnwidth]{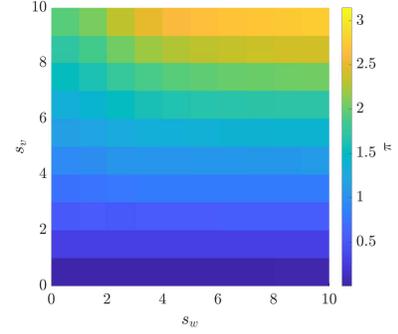}
       \caption{Results of the sensitivity analysis with respect to the endogenous ($s_v$) and exogenous ($s_w$) control costs. The heat-map represents the value of the index $\pi$, defined in \eqref{eq:pi}, for the different values of the control cost parameters. }\label{fig:weights}
\end{figure}

\subsection{Myopic vs predictive policies}\label{sec:optimalT}

Next, we consider the problem of choosing the optimal length for the prediction horizon. This is a critical choice that a policy maker has to perform. In fact, if $T_h$ is too small the policies will be myopic failing to prepare in time for the future evolution of the epidemics. On the other hand, longer prediction horizons will require an increased computational effort. Moreover, due to the difference between the real dynamics and the ones used in \eqref{eq:opt_problem}, a long prediction horizon may lead to inaccurate estimation and consequently to  incorrect precautionary policies, see Figure~\ref{fig:nl_l}. Therefore, we believe that there is an optimal length for the prediction horizon. 

\begin{figure}[t]
\centering    
  \subfloat[]{  \includegraphics[width=0.48\columnwidth]{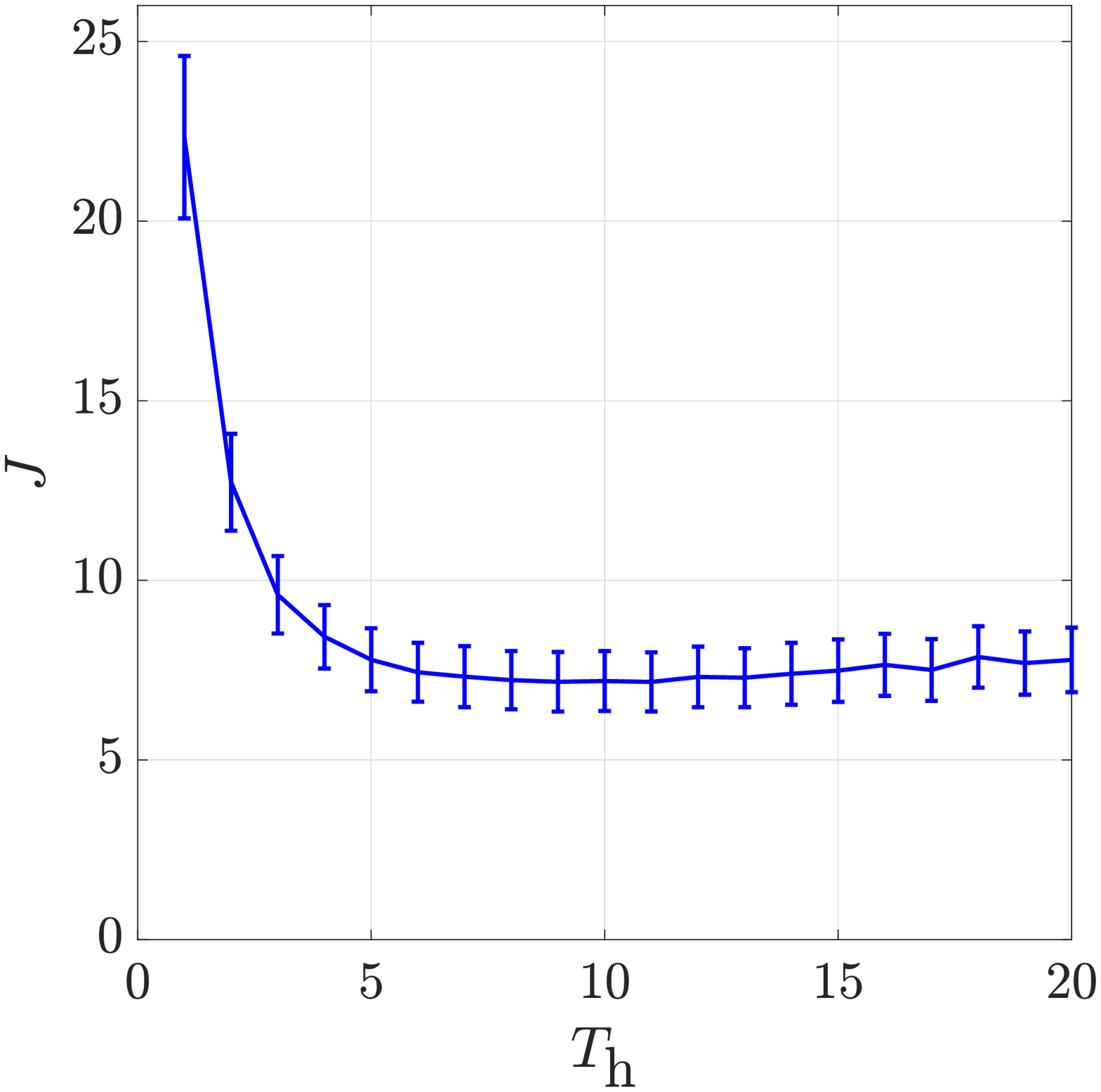}\label{fig:cost}} \hfill 
  \subfloat[]{\includegraphics[width=0.48\columnwidth]{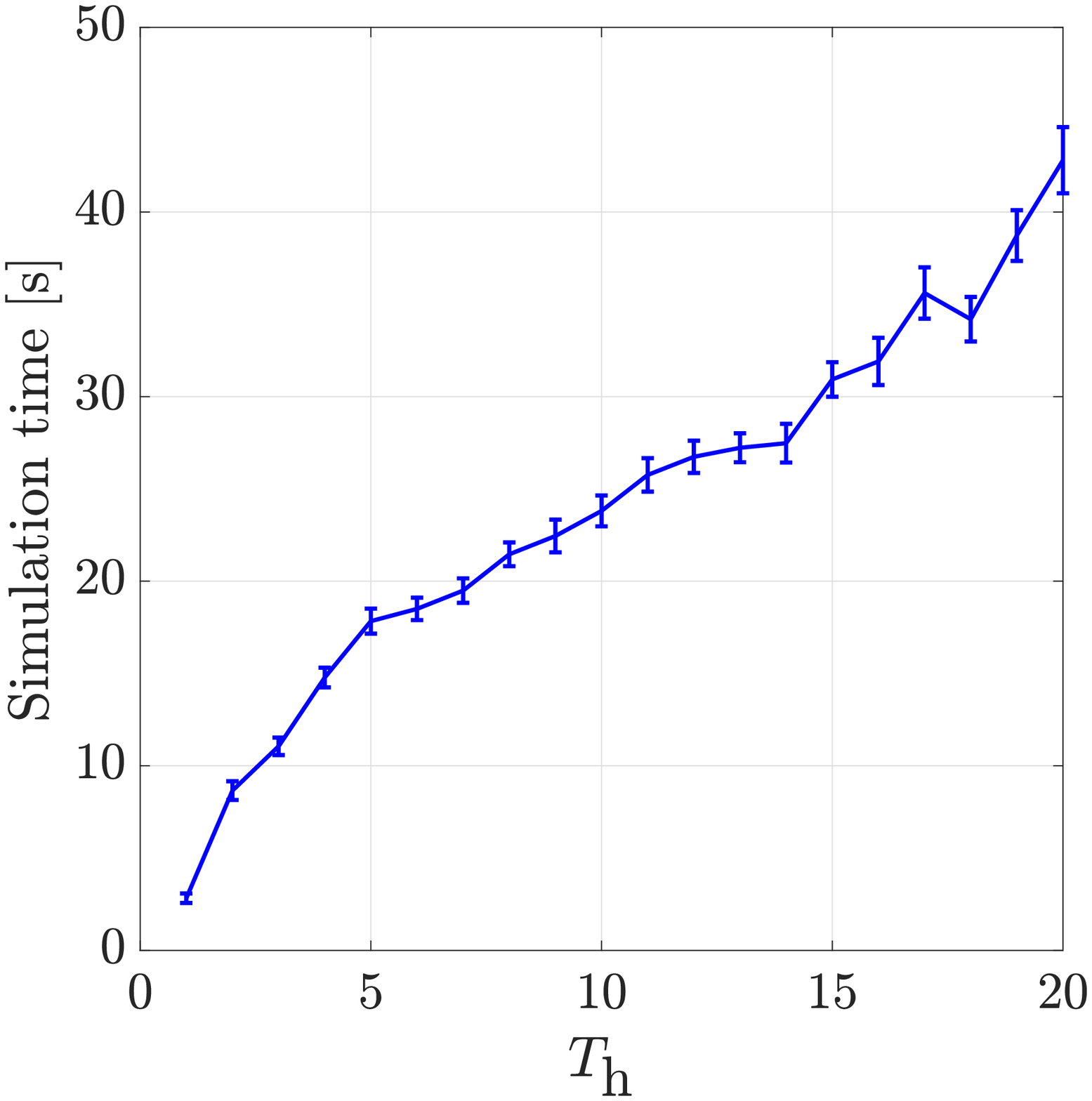} \label{fig:time}}
       \caption{In (a), we illustrate the value of $J$, i.e., the sum of the cost over the whole simulations, as the prediction horizon $T_{\tup h}$ grows. In (b), we show the computational time for one simulation for growing values of $T_{\tup h}$. In both figures, the vertical bars represent $95\%$ confidence intervals  computed over $50$ independent realizations of the random variables $\overline{A}$, $\overline{\beta}$, and $\bs{\hat{x}}$.}\label{fig:cost_and_time}
\end{figure}

To validate this claim via simulations, we generate $50$ different scenarios in which $\overline{A}$, $\overline{\beta}$ and ${\bs \hat x}$ are randomly chosen as described in the previous section, while the rest of the coefficients of the problem do not change. As index of the performance we consider the sum over the whole simulation of the cost actually experienced by the communities, we denote it simply by $J$. Figure~\ref{fig:cost} clearly depicts a Pareto front in which, after an initial reduction of $J$ due to the growing $T_{\tup h}$, there is a diminishing return. Moreover, for even greater values of $T_{\tup h}$ the performance starts to slowly get worse. The computational time of each one of the simulation grows linearly with the value of $T_{\tup{h}}$, as shown in Figure~\ref{fig:time}, this can be a serious issue in the case of large scale networks. Therefore, the optimal value of $T_{\tup{h}}$ for the simulated cases is in the range $5-7$. In fact, it is large enough to benefit from the initial steep increment in terms of performance but not that large to make the estimations less accurate.

\section{Conclusion and outlook}\label{sec:conclusions}
In this paper, the problem of designing optimal \gls{NPI} to ensure containing an epidemic on a network while minimizing the costs due to restrictive policies has been cast as a nonlinear constrained optimization problem over a prediction horizon. We have proposed a solution via a linearization technique and a receding horizon strategy, in which the real  evolution of the disease is used as a feedback to cast a new instance of the optimization problem and compute the upcoming  policies that have to be implemented.

The performance of the proposed control scheme has been demonstrated via preliminary numerical simulations. Specifically, the good accuracy of the linearized dynamics for sufficiently long time horizons guarantees accurate farsighted forecasts, avoiding thus the risks of relying on myopic policies. Our numerical findings have shown nontrivial solutions for the optimal control strategies, depending on the costs associated with the implementation of the different types of \gls{NPI} and on the requirements on the desired evolution of the epidemic outbreak. We finally would like to stress that, while the current formulation has been developed using an \gls{SIS} epidemic model, the same control scheme can be extended to more complex epidemic models, such as those used to capture the features of the ongoing COVID-19 pandemic~\cite{Giordano2020}. Such an extension, and the analysis of the proposed control scheme in real-world scenarios is an important direction for the future research.

\bibliographystyle{IEEEtran}
\bibliography{ref_epi}

\begin{thebibliography}{10}
\providecommand{\url}[1]{#1}
\csname url@samestyle\endcsname
\providecommand{\newblock}{\relax}
\providecommand{\bibinfo}[2]{#2}
\providecommand{\BIBentrySTDinterwordspacing}{\spaceskip=0pt\relax}
\providecommand{\BIBentryALTinterwordstretchfactor}{4}
\providecommand{\BIBentryALTinterwordspacing}{\spaceskip=\fontdimen2\font plus
\BIBentryALTinterwordstretchfactor\fontdimen3\font minus
  \fontdimen4\font\relax}
\providecommand{\BIBforeignlanguage}[2]{{%
\expandafter\ifx\csname l@#1\endcsname\relax
\typeout{** WARNING: IEEEtran.bst: No hyphenation pattern has been}%
\typeout{** loaded for the language `#1'. Using the pattern for}%
\typeout{** the default language instead.}%
\else
\language=\csname l@#1\endcsname
\fi
#2}}
\providecommand{\BIBdecl}{\relax}
\BIBdecl

\bibitem{Haug2020}
N.~Haug \emph{et~al.}, ``Ranking the effectiveness of worldwide {COVID-19}
  government interventions,'' \emph{Nature Human Behavious}, vol.~4, no.~12,
  pp. 1303--1312, 2020.

\bibitem{prem2020effect}
K.~Prem \emph{et~al.}, ``{The effect of control strategies to reduce social
  mixing on outcomes of the {COVID}-19 epidemic in {Wuhan}, {China}: a
  modelling study},'' \emph{The Lancet Public Health}, vol.~5, no.~5, pp.
  e261--e270, 2020.

\bibitem{Bonaccorsi2020economic}
G.~Bonaccorsi \emph{et~al.}, ``Economic and social consequences of human
  mobility restrictions under {COVID}-19,'' \emph{Proceedings of the National
  Academy of Sciences}, vol. 117, no.~27, pp. 15\,530--15\,535, 2020.

\bibitem{Nowzari2016}
C.~{Nowzari}, V.~M. {Preciado}, and G.~J. {Pappas}, ``{Analysis and Control of
  Epidemics: A Survey of Spreading Processes on Complex Networks},'' \emph{IEEE
  Control Systems Magazine}, vol.~36, no.~1, pp. 26--46, 2016.

\bibitem{Mei2017}
W.~Mei, S.~Mohagheghi, S.~Zampieri, and F.~Bullo, ``On the dynamics of
  deterministic epidemic propagation over networks,'' \emph{Annual Reviews in
  Control}, vol.~44, pp. 116--128, 2017.

\bibitem{Pare2020review}
P.~E. Paré, C.~L. Beck, and T.~Başar, ``Modeling, estimation, and analysis of
  epidemics over networks: An overview,'' \emph{Annual Reviews in Control},
  vol.~50, pp. 345--360, 2020.

\bibitem{zino:2021:epidemics_survey}
L.~{Zino} and M.~{Cao}, ``{Analysis, Prediction, and Control of Epidemics: A
  Survey from Scalar to Dynamic Network Models},'' \emph{IEEE Circuits and
  Systems Magazine}, 2021, to appear (arXiv:2103.00181).

\bibitem{ye:cenedese:2020:epidemics}
M.~Ye, J.~Liu, C.~Cenedese, Z.~Sun, and M.~Cao, ``A network sis meta-population
  model with transportation flow,'' \emph{IFAC-PapersOnLine}, vol.~53, no.~2,
  pp. 2562--2567, 2020, 21st IFAC World Congress.

\bibitem{Giordano2020}
G.~Giordano \emph{et~al.}, ``{Modelling the {COVID-19} epidemic and
  implementation of population-wide interventions in Italy},'' \emph{Nature
  Medicine}, vol.~26, no.~6, pp. 855--860, 2020.

\bibitem{Casella2020}
F.~{Casella}, ``{Can the COVID-19 Epidemic Be Controlled on the Basis of Daily
  Test Reports?}'' \emph{IEEE Control Systems Letters}, vol.~5, no.~3, pp.
  1079--1084, 2021.

\bibitem{Calafiore2020}
G.~C. Calafiore, C.~Novara, and C.~Possieri, ``{A time-varying SIRD model for
  the COVID-19 contagion in Italy},'' \emph{Annual Reviews in Control},
  vol.~50, pp. 361--372, 2020.

\bibitem{Gatto2020}
M.~Gatto \emph{et~al.}, ``{Spread and dynamics of the COVID-19 epidemic in
  Italy: Effects of emergency containment measures},'' \emph{Proceedings of the
  National Academy of Sciences}, vol. 117, no.~19, pp. 10\,484--10\,491, 2020.

\bibitem{dellarossa2020}
F.~{Della Rossa} \emph{et~al.}, ``{A network model of Italy shows that
  intermittent regional strategies can alleviate the {COVID}-19 epidemic},''
  \emph{Nature Communications}, vol.~11, no.~1, 2020.

\bibitem{Carli2020}
R.~Carli, G.~Cavone, N.~Epicoco, P.~Scarabaggio, and M.~Dotoli, ``{Model
  predictive control to mitigate the COVID-19 outbreak in a multi-region
  scenario},'' \emph{Annual Reviews in Control}, vol.~50, pp. 373 -- 393, 2020.

\bibitem{Parino2021}
F.~Parino, L.~Zino, M.~Porfiri, and A.~Rizzo, ``{Modelling and predicting the
  effect of social distancing and travel restrictions on COVID-19 spreading},''
  \emph{Journal of the Royal Society Interface}, p. 20200875, 2021.

\bibitem{kohler:2020:robust_MPC_COVID}
J.~Köhler, L.~Schwenkel, A.~Koch, J.~Berberich, P.~Pauli, and F.~Allgöwer,
  ``{Robust and optimal predictive control of the COVID-19 outbreak},''
  \emph{Annual Reviews in Control}, vol.~51, pp. 525--539, 2021.

\bibitem{Bin2021}
M.~Bin \emph{et~al.}, ``Post-lockdown abatement of covid-19 by fast periodic
  switching,'' \emph{PLOS Computational Biology}, vol.~17, no.~1, pp. 1--34, 01
  2021.

\bibitem{Parino2021vaccine}
F.~Parino, L.~Zino, G.~C. Calafiore, and A.~Rizzo, ``A model predictive control
  approach to optimally devise a two-dose vaccination rollout: A case study on
  {COVID-19} in {Italy},'' \emph{International Journal of Robust and Nonlinear
  Control}, 2021.

\bibitem{Preciado2014}
V.~M. {Preciado}, M.~{Zargham}, C.~{Enyioha}, A.~{Jadbabaie}, and G.~J.
  {Pappas}, ``{Optimal Resource Allocation for Network Protection Against
  Spreading Processes},'' \emph{IEEE Transactions on Control of Network
  Systems}, vol.~1, no.~1, pp. 99--108, 2014.

\bibitem{ramirez2018}
E.~Ramírez-Llanos and S.~Martínez, ``Distributed discrete-time optimization
  algorithms with applications to resource allocation in epidemics control,''
  \emph{Optimal Control Applications \& Methods}, vol.~39, no.~1, pp. 160--180,
  2018.

\bibitem{Mai2018}
V.~S. {Mai}, A.~{Battou}, and K.~{Mills}, ``{Distributed Algorithm for
  Suppressing Epidemic Spread in Networks},'' \emph{IEEE Control Systems
  Letters}, vol.~2, no.~3, pp. 555--560, 2018.

\bibitem{Eshghi2016}
S.~{Eshghi}, M.~H.~R. {Khouzani}, S.~{Sarkar}, and S.~S. {Venkatesh},
  ``{Optimal Patching in Clustered Malware Epidemics},'' \emph{IEEE/ACM
  Transactions on Networking}, vol.~24, no.~1, pp. 283--298, 2016.

\bibitem{Grandits2019}
P.~Grandits, R.~M. Kovacevic, and V.~M. Veliov, ``{Optimal control and the
  value of information for a stochastic epidemiological SIS-model},''
  \emph{Journal of Mathematical Analysis and Applications}, vol. 476, no.~2,
  pp. 665 -- 695, 2019.

\bibitem{selley2015}
F.~S{\'e}lley, {\'A}.~Besenyei, I.~Z. Kiss, and P.~L. Simon, ``{Dynamic Control
  of Modern, Network-Based Epidemic Models},'' \emph{SIAM Journal on Applied
  Dynamical Systems}, vol.~14, no.~1, pp. 168--187, 2015.

\bibitem{Kohler2018}
J.~{Köhler}, C.~{Enyioha}, and F.~{Allgöwer}, ``{Dynamic Resource Allocation
  to Control Epidemic Outbreaks A Model Predictive Control Approach},'' in
  \emph{2018 Annual American Control Conference}, 2018, pp. 1546--1551.

\bibitem{Watkins2020}
N.~J. {Watkins}, C.~{Nowzari}, and G.~J. {Pappas}, ``{Robust Economic Model
  Predictive Control of Continuous-Time Epidemic Processes},'' \emph{IEEE
  Transactions on Automatic Control}, vol.~65, no.~3, pp. 1116--1131, 2020.

\bibitem{MORATO2020417}
M.~M. Morato, S.~B. Bastos, D.~O. Cajueiro, and J.~E. Normey-Rico, ``{An
  optimal predictive control strategy for COVID-19 (SARS-CoV-2) social
  distancing policies in Brazil},'' \emph{Annual Reviews in Control}, vol.~50,
  pp. 417--431, 2020.

\bibitem{Pare2020}
P.~E. {Paré}, J.~{Liu}, C.~L. {Beck}, B.~E. {Kirwan}, and T.~{Başar},
  ``{Analysis, Estimation, and Validation of Discrete-Time Epidemic
  Processes},'' \emph{IEEE Transactions on Control Systems Technology},
  vol.~28, no.~1, pp. 79--93, 2020.

\bibitem{Liu2020}
F.~Liu, S.~CUI, X.~Li, and M.~Buss, ``On the stability of the endemic
  equilibrium of a discrete-time networked epidemic model,'' vol.~53, no.~2,
  2020, pp. 2576--2581, 21st IFAC World Congress.

\bibitem{nocedal2006sequential}
J.~Nocedal and S.~J. Wright, ``Sequential quadratic programming,''
  \emph{Numerical optimization}, pp. 529--562, 2006.

\bibitem{schittkowski:1986:NLPQL}
K.~Schittkowski, ``{NLPQL: A fortran subroutine solving constrained nonlinear
  programming problems},'' \emph{Annals of Operations Research}, vol.~5, pp.
  485--500, 1986.

\end{thebibliography}

%
%

\end{document}